\documentclass[
  aps,
  prd,
  reprint,
  nobalancelastpage,
  superscriptaddress,
  nofootinbib,
  longbibliography
]{revtex4-2}

\usepackage[T1]{fontenc}
\usepackage{amsmath,amssymb}
\usepackage{mathtools}
\usepackage{bm}
\usepackage{amsthm}
\usepackage{graphicx}
\usepackage{xcolor}
\usepackage[colorlinks=true,allcolors=blue]{hyperref}
\usepackage{cleveref}

\newtheorem{theorem}{Theorem}
\newcommand{\id}{\mathbf{1}}
\newcommand{\Tr}{\operatorname{Tr}}
\newcommand{\ii}{\mathrm{i}}
\newcommand{\e}{\mathrm{e}}
\newcommand{\cH}{\mathcal{H}}
\newcommand{\cL}{\mathcal{L}}

\newcommand{\dd}{\mathrm{d}}

\makeatletter
\renewcommand{\bibsection}{%
  \par
  \addvspace{12\p@}%
  \noindent
  \makebox[\columnwidth][c]{%
    \rule{0.55\columnwidth}{0.5\p@}%
  }%
  \par
  \nobreak
  \addvspace{10\p@}%
}
\makeatother

\begin{document}

\title{Inequalities and Positivity in Modular Flowed Entanglement Entropy}

\author{Liangyu Chen}
\email{liangyu-chen@mail.tsinghua.edu.cn}
\affiliation{Yau Mathematical Sciences Center, Tsinghua University, Beijing 100084, China}
\date{\today}

\begin{abstract}
The modular Hamiltonian plays an important role in quantum information theory, quantum field theory, and the AdS/CFT correspondence. In this paper, we study a specific dynamical setup: starting from the tensor product of the reduced density matrices of a reference state, we evolve this product state under the modular flow generated by the reference state itself. We then investigate the resulting change in the entanglement entropy of a given region, which we call the modular entanglement shift (MES). Unlike evolution under a fixed quantum channel, this modular evolution obeys no general principle requiring the MES to have a definite sign. For a bipartition into subsystems A and B, we establish an exact entropy balance: the sum of the MESs for A and B equals the mutual information generated along the orbit. Consequently, their sum is nonnegative, even though either individual shift may be negative. For two disjoint intervals in the vacuum of a two-dimensional conformal field theory, however, we prove that a global conformal involution exchanging the intervals forces the two MESs to be equal. It follows that each MES is nonnegative along the modular flow. We establish this result both using a conformally transported regulator and intrinsically in terms of Araki relative entropy. Finally, we identify the modular-flow orbit with a Connes cocycle orbit and discuss possible applications to the AdS/CFT correspondence.
\end{abstract}

\maketitle

\section{Introduction}

For a faithful density matrix $\rho$, the modular Hamiltonian is
$K_\rho=-\log\rho$; in QFT, where local algebras are generically type III,
the intrinsic object is the Tomita--Takesaki modular automorphism group
\cite{Takesaki:1970mod,Haag:1996hvx,Witten:2018zxz}.  Vacuum modular flow is
geometric for a Rindler half-space and, by conformal transport, for a ball
and related 2D geometries \cite{Bisognano:1975ih,Bisognano:1976za,
Hislop:1982mod,Casini:2011kv,Cardy:2016eh,Arias:2016local}.  For generic states
or disconnected regions the generator is nonlocal, as exact multi-interval
examples demonstrate \cite{Casini:2009multi,Longo:2009geo,Wong:2018glue,
Arias:2018scalar,Hollands:2019multi}.  This structure also drives many-body
studies of entanglement Hamiltonians \cite{Dalmonte:2022rlo,Kokail:2020eht}.

Perturbations of geometric examples yield stress-tensor expressions for
shape-deformed modular Hamiltonians and entropic proofs of averaged energy
conditions \cite{Faulkner:2015csl,Faulkner:2016mzt}, with a recent nonlocal
extension to deformed null cuts \cite{Lu:2025eqy}.  The exact Markov property
on null planes \cite{Casini:2017roe} underlies the QNEC, proposed through
quantum focusing and subsequently proved \cite{Bousso:2015qfc,
Bousso:2015qne,Balakrishnan:2017bjg}, and connects it to half-sided modular
inclusions \cite{Wiesbrock:1993hsm,Ceyhan:2018zfg}.  For disconnected
regions, small-cross-ratio expansions and free-fermion transport provide
perturbative and exact control \cite{Cardy:2013mi,Chandrasekaran:2021tkb,
Agon:2015mi,Chen:2022cmr,Chen:2022mpt}.

In holography, JLMS proposal and quantum error correction relate boundary and bulk
modular generators and support entanglement-wedge reconstruction
\cite{Faulkner:2013ana,Jafferis:2014lza,Jafferis:2015del,Almheiri:2014lwa,
Dong:2016eik,Harlow:2016vwg,Faulkner:2017vdd,Cotler:2017rec,
Faulkner:2018faa}.  Nonlocal flow probes the causal shadow, while modular
transport and chaos encode bulk geometry \cite{Chen:2022eyi,Czech:2019vih,
DeBoer:2019kdj,Czech:2023zmq}.  Relative modular operators and Connes
cocycles supply a regulator-independent language in QFT and gravity
\cite{Connes:1973coc,Bousso:2020yxi,Leutheusser:2021frk,
Witten:2021unn}. 

In this paper, we ask how local entanglement changes when a decorrelated
state is transported by a correlated modular Hamiltonian. The reference state $\tau_{AB}$ is stationary under its own flow.  We retain
its marginals and modular generator but remove its initial correlations,
evolving $\tau_A\otimes\tau_B$ to measure how the nonlocal part of $K_\tau$
rebuilds correlations and redistributes entropy.  The construction is finite
in type I and extends to type III through Umegaki and Araki relative entropy
\cite{Umegaki:1962rel,Araki:1975rel}.

Let $\tau_{AB}>0$ be a faithful reference state and
$\sigma_{AB}=\tau_A\otimes\tau_B$ the product of its marginals.  We study
the modular orbit
\begin{equation}
 \omega_{AB}(s)
 =\e^{-\ii sK_\tau}\sigma_{AB} \, \e^{\ii sK_\tau},
 \qquad K_\tau=-\log\tau_{AB},
 \label{eq:main-orbit}
\end{equation}
and define the modular entanglement shift (MES)
\begin{equation}
 \Delta S_X(s):=S(\omega_X(s))-S(\tau_X),
 \qquad X=A,B.
 \label{eq:main-mes}
\end{equation}
Although Eq.~\eqref{eq:main-orbit} resembles a Stinespring dilation
\cite{Stinespring:1955cp}, the
assignment from $\tau_{AB}$ to $\omega_A(s)$ is not a fixed quantum channel:
the input, environment, and generator all depend on the reference state.
The usual entropy theorem for unital channels therefore does not determine
the sign of Eq.~\eqref{eq:main-mes}.

We establish three results.  First, the two local MES values obey an exact
balance law: their sum equals the mutual information generated along the
orbit, although either individual shift may be negative.  At short modular
time, the curvature of this total shift defines a positive
Bogoliubov--Kubo--Mori modular susceptibility.  Second, any pair of disjoint
vacuum intervals on the compactified line admits a M\"obius involution that
exchanges them.  Vacuum covariance then equates the two shifts: with a
conformally adapted type-I regulator, each is separately non-negative, while
the intrinsic continuum statement is expressed through the algebraic MES
defined by Araki relative entropy.  Third,
Eq.~\eqref{eq:main-orbit} is precisely the state orbit generated by the
Connes cocycle of $\tau_{AB}$ relative to $\sigma_{AB}$.

\section{Entropy balance along the modular flow}
\label{sec:main-balance}
Since the modular flow is a unitary operation for the whole $AB$ system, so the entropy of $AB$ is invariant along the modular flow. Notice that the initial state $\sigma_{AB}$is a product state. We immediately get the following inequality.
\begin{theorem}[Entropy balance]
For every real modular time $s$,
\begin{equation}
 \Delta S_A(s)+\Delta S_B(s)
 =I(A:B)_{\omega(s)}\geq0.
 \label{eq:main-balance}
\end{equation}
\end{theorem}
The proof is simple, the unitary invariance and additivity imply
$S(\omega_{AB}(s))=S(\sigma_{AB})=S(\tau_A)+S(\tau_B)$; inserting this into
$I(A:B)_\omega=S(\omega_A)+S(\omega_B)-S(\omega_{AB})$ proves
Eq.~\eqref{eq:main-balance}. And the equality condition for $I(A:B)_{\omega(s)} =0$ is:
\begin{equation}
 \omega_{AB}(s)=\omega_A(s)\otimes\omega_B(s).
 \label{eq:main-balance-saturation}
\end{equation}
From this inequality, we see that the local purification $\Delta S_{A}(s) <0 $ is possible only if it
is overcompensated by entropy growth in the complementary subsystem. This is very different from the usual quantum channel. If a quantum channel is unital, then one always has $S(\Phi(\rho)) \ge S(\rho)$. For fixed reference state $\tau_{AB}$ and modular time $s$,
the reduced map
\begin{equation}
 \Phi^{(\tau)}_{A,s}(\rho_A)
 =\Tr_B\!\left[
 \e^{-\ii sK_\tau}(\rho_A\otimes\tau_B) \, \e^{\ii sK_\tau}
 \right]
 \label{eq:fixed-modular-channel}
\end{equation}
is indeed CPTP, but changing the reference state $\tau_{AB}$ simultaneously
changes the input $\tau_A\otimes\tau_B$, the environment $\tau_B$, and the
generator $K_\tau$.  Hence the assignment from reference data to
modular-flowed data is not described by a single quantum channel. 
The exact local criterion follows instead from relative entropy.  Define
$K_A=-\log\tau_A$ and
\begin{equation}
 \Delta\langle K_A\rangle_s
 :=\Tr[(\omega_A(s)-\tau_A)K_A].
 \label{eq:sm-local-energy}
\end{equation}
Expanding
$D(\omega_A(s)\Vert\tau_A)
=\Tr[\omega_A(s)(\log\omega_A(s)-\log\tau_A)]$ gives \cite{Blanco:2013rel,Blanco:2017mei}
\begin{equation}
 \Delta S_A(s)=\Delta\langle K_A\rangle_s
 -D(\omega_A(s)\Vert\tau_A).
 \label{eq:sm-local-identity}
\end{equation}
Consequently,
\begin{equation}
 \Delta S_A(s)<0
 \quad\Longleftrightarrow\quad
 \Delta\langle K_A\rangle_s
 <D(\omega_A(s)\Vert\tau_A).
 \label{eq:sm-purification-criterion}
\end{equation}

Notice that the sum of two MES is a mutual information, then we  can apply the standard BKM susceptibility technique to the two local entropy changes. The first order perturbative around $s=0$ is zero, $\Delta S_A'(0)=\Delta S_B'(0)=0$.  Its leading correlation is
quadratic.  Let
\begin{equation}
 X=-\ii[K_\tau,\sigma_{AB}],\quad
 X_{\rm corr}=X-X_A\otimes\tau_B-\tau_A\otimes X_B,
\end{equation}
where $X_A=\Tr_BX$ and $X_B=\Tr_AX$, then
\begin{align}
 I(A:B)_{\omega(s)}
 &=\frac{s^2}{2}\,\chi_{\rm mod}(\tau)+O(s^3),
 \label{eq:main-susceptibility}\\
 \chi_{\rm mod}(\tau)
 &=\Tr[X_{\rm corr}\mathcal T_\sigma(X_{\rm corr})]\geq0.
 \label{eq:main-susceptibility-definition}
\end{align}
Here $\mathcal T_\rho$ is the Fr\'echet derivative of $\log\rho$.

\section{An exactly solvable two-qubit orbit}
\label{sec:main-qubit}
Let us consider a two-qubit model with the following un-normalized modular Hamiltonian
\begin{equation}
 \begin{aligned}
  H&=a \, Z\otimes\id+b\, \id\otimes Z
  +g \, (X\otimes X-Y\otimes Y),\\
  \tau_{AB}&=\frac{\e^{-H}}{\Tr\e^{-H}}.
 \end{aligned}
 \label{eq:main-qubit-H}
\end{equation}
where $X,Y,Z$ denote the Pauli matrices. Below we will work in the computational basis
$\{|00\rangle,|01\rangle,|10\rangle,|11\rangle\}$. The matrix form of $H$ is
\begin{equation}
 H(a,b,g)=
 \begin{pmatrix}
 u&0&0&2g\\
 0&v&0&0\\
 0&0&-v&0\\
 2g&0&0&-u
 \end{pmatrix},
 \label{eq:qubit-general-H}
\end{equation}
 where
\begin{equation}
 u=a+b,\quad v=a-b
 \label{eq:qubit-uvr}
\end{equation}
We will see that both signs of the local MES already occur in this two-qubit model. Set $r=(u^2+4g^2)^{1/2}$, the marginals are diagonal,
$\tau_A=\operatorname{diag}(\alpha,1-\alpha)$ and
$\tau_B=\operatorname{diag}(\beta,1-\beta)$, with
\begin{equation}
 \begin{aligned}
  \alpha&=\frac{A_r+\e^{-v}}{Z_\tau},&
  \beta&=\frac{A_r+\e^{v}}{Z_\tau},\\
  A_r&=\cosh r-\frac{u}{r}\sinh r,&
  Z_\tau&=2(\cosh r+\cosh v).
 \end{aligned}
 \label{eq:main-alpha-beta}
\end{equation}
The initial product state is $\sigma_{AB}
 =\operatorname{diag}(p_{00},p_{01},p_{10},p_{11})$ with
\begin{align}
p_{00}&=\alpha\beta,&p_{01}&=\alpha(1-\beta),\nonumber\\
 p_{10}&=(1-\alpha)\beta,&p_{11}&=(1-\alpha)(1-\beta).
 \label{eq:main-qubit-product-populations}
\end{align}
The modular flowed state becomes 
\begin{equation}
 \omega_{AB}(s)=
 \begin{pmatrix}
  p_{00}+\delta&0&0&\gamma\\
  0&p_{01}&0&0\\
  0&0&p_{10}&0\\
  \gamma^*&0&0&p_{11}-\delta
 \end{pmatrix},
 \label{eq:main-qubit-density-matrix}
\end{equation}
where
\begin{equation}
 \begin{aligned}
  \delta&=\frac{4g^2}{r^2}\sin^2(rs)(p_{11}-p_{00}),\\
  \gamma&=\ii\frac{2g}{r}\sin(rs) \Big(\cos(rs)-\ii\frac{u}{r}\sin(rs)\Big)(p_{00}-p_{11}).
 \end{aligned}
 \label{eq:main-qubit-delta-gamma}
\end{equation}
Taking the two partial traces and 
writing $h_2(x) \equiv -x\log x-(1-x)\log(1-x)$, we obtain the two MES are 
\begin{equation}
\begin{aligned}
    \Delta S_A &= h_2(\alpha+\delta)-h_2(\alpha), \\
 \Delta S_B &= h_2(\beta+\delta)-h_2(\beta).
\end{aligned}  \label{eq:main-qubit-mes}
\end{equation}

\begin{figure*}[t]
 \centering
 \includegraphics[width=0.46\textwidth]{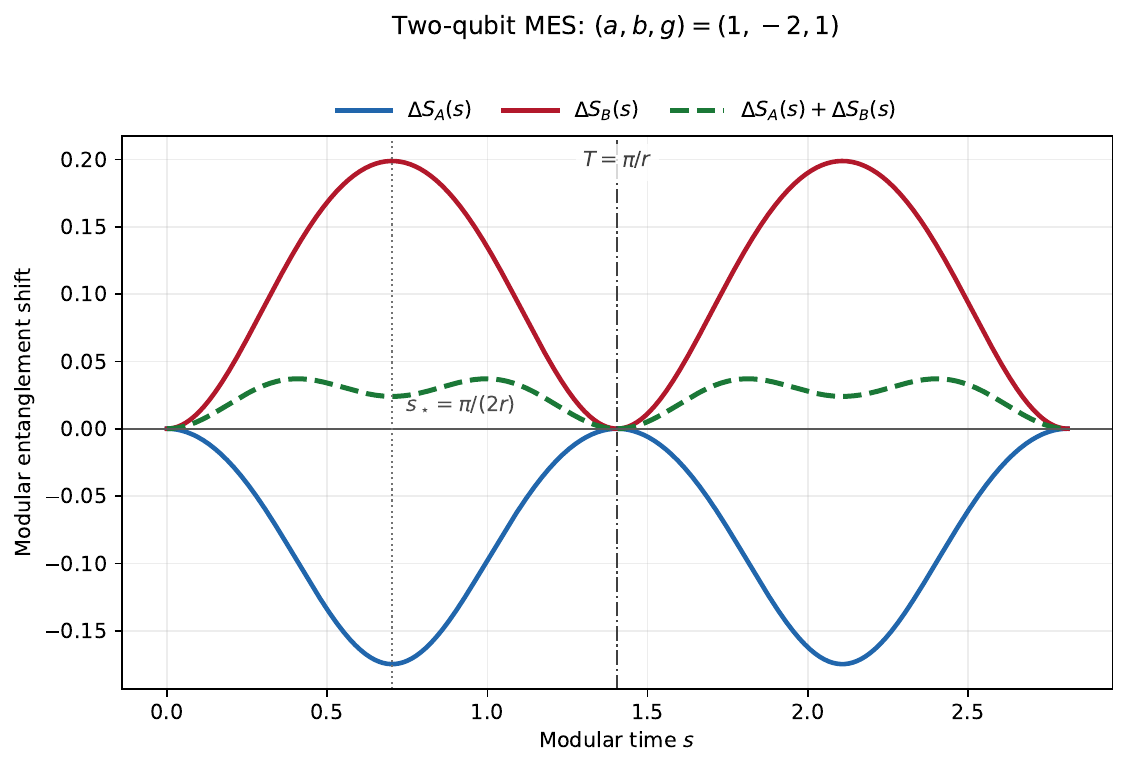}
 \hfill
 \includegraphics[width=0.46\textwidth]{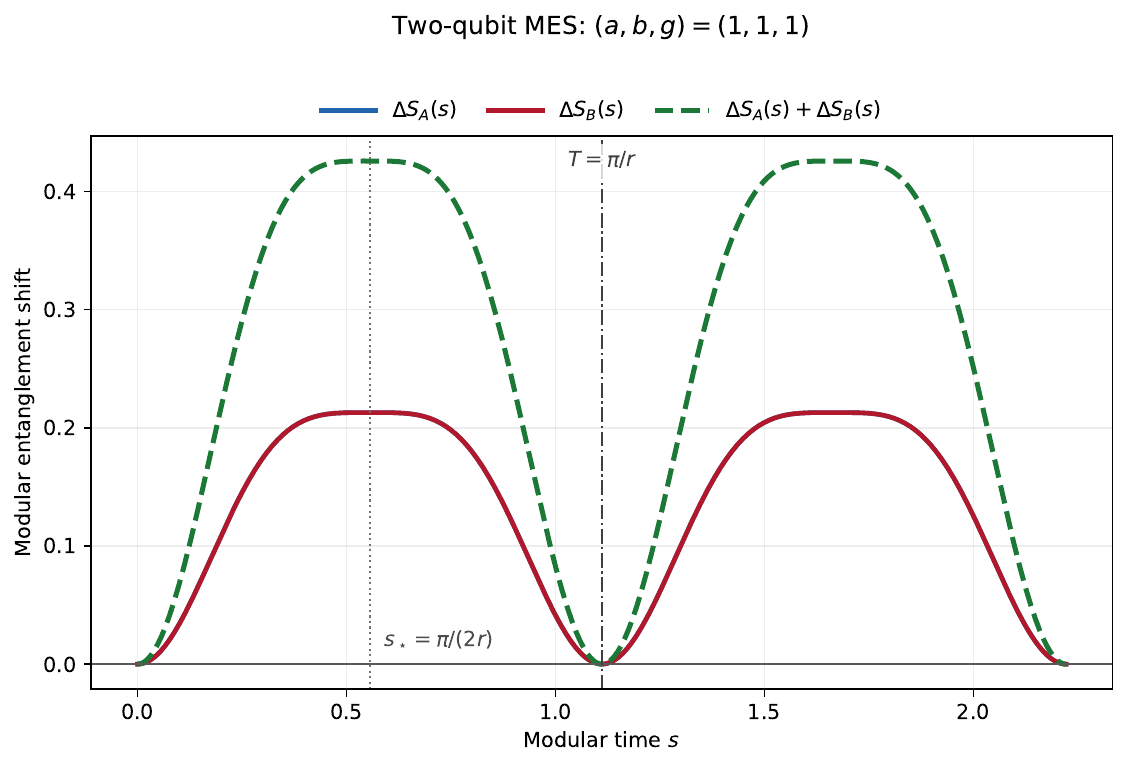}
 \caption{Modular-time dependence of the two local MES values and their
 sum.  Left: $(a,b,g)=(1,-2,1)$ exhibits local purification,
 $\Delta S_A<0<\Delta S_B$.  Right: for the exchange-symmetric point
 $(a,b,g)=(1,1,1)$, the two shifts coincide and are non-negative.  The
 dashed curve is $I(A:B)_{\omega(s)}$ by Eq.~\eqref{eq:main-balance}; the
 dotted and dash-dotted lines mark $s_\star=T/2$ and $T=\pi/r$.}
 \label{fig:main-qubit}
\end{figure*}

At the asymmetric point $(a,b,g)=(1,-2,1)$ and
$s_\star=\pi/(2\sqrt5)$,
\begin{equation}
 \Delta S_A\simeq-0.174826,\qquad
 \Delta S_B\simeq0.198799,
 \label{eq:main-asym-values}
\end{equation}
while their sum is $0.023973>0$.  Negative local MES is therefore not an
exceptional or perturbative effect.  At the symmetric point $a=b=g=1$,
exchange symmetry gives $\alpha=\beta$ and hence
\begin{equation}
 \Delta S_A(s)=\Delta S_B(s)
 =\frac12 I(A:B)_{\omega(s)}\geq0.
 \label{eq:main-qubit-symmetric}
\end{equation}
At $s_\star=\pi/(4\sqrt2)$ both marginals are maximally mixed and each MES
is approximately $0.212912$.  

The modular orbits of these two cases are plotted in Fig.\ref{fig:main-qubit}
The orbit has period $T=\pi/r$, this periodicity follows directly from
\cref{eq:main-qubit-delta-gamma}.  Both $\delta(s)$ and $\gamma(s)$ have period
\begin{equation}
 T=\frac{\pi}{r},
 \label{eq:qubit-recurrence-period}
\end{equation}
and hence
\begin{equation}
 \begin{aligned}
  \omega_{AB}(s+T)
  &=\omega_{AB}(s),\\
  \Delta S_X(s+T)
  &=\Delta S_X(s),
  \qquad X=A,B.
 \end{aligned}
 \label{eq:qubit-mes-periodicity}
\end{equation}
The recurrence has a simple spectral origin.  In the eigenbasis of
$K_\tau$, the modular evolution is governed by phases
$\e^{-\ii(\kappa_m-\kappa_n)s}$.  In the present model, the only active
modular gap is the level splitting $2r$ of the even-parity block, and hence
\begin{equation}
 T=\frac{2\pi}{2r}=\frac{\pi}{r}.
\end{equation}

\section{MES positivity for two  intervals in a 2D CFT vacuum}
\label{sec:main-cft}
Exchange symmetry has a nontrivial continuum realization.  Let
$A=[x_1,x_2]$ and $B=[x_3,x_4]$ be disjoint intervals on the compactified
line, ordered as $x_1<x_2<x_3<x_4$.  A real M\"obius transformation $g$
maps their endpoints to
\begin{equation}
 (x_1,x_2,x_3,x_4)\xmapsto{g}(-1,-q,q,1),
 \qquad 0<q<1.
 \label{eq:main-canonical-pair}
\end{equation}
For $u=x_{12}x_{34}/(x_{13}x_{24})$, the canonical parameter is
$q=(1-\sqrt u)/(1+\sqrt u)$.  Thus the exchange exists for every
cross-ratio, not only for equal-length or reflection-symmetric intervals.
In this frame the global conformal map $j_q(x)=-q/x$ exchanges the two
intervals.  Therefore, we have
\begin{equation}
 \begin{aligned}
  h&=g^{-1}\circ j_q\circ g,& h^2&=\mathrm{id},\\
  h(A)&=B,& h(B)&=A,
 \end{aligned}
 \label{eq:main-involution}
\end{equation}
Now, we will use this conformal  involution to prove the positivity of each MES. In a conformal adapted type-I regulator case,  the argument is immediate. Let $U(h)$ denote the unitary implementing the conformal
involution. Since $h$ preserves $A\cup B$ and the vacuum is invariant under
global conformal transformations,
\begin{equation}
  U(h)\tau_{AB}U(h)^\dagger=\tau_{AB}.
\end{equation}
It follows that $U(h)$ commutes with
$K_\tau=-\log\tau_{AB}$, and hence with the reference modular flow.
Moreover, $U(h)$ exchanges the two tensor factors and maps $\tau_A$ to
$\tau_B$, so that it leaves
$\sigma_{AB}=\tau_A\otimes\tau_B$ invariant. Consequently, the flowed
marginals $\omega_A(s)$ and $\omega_B(s)$ are unitarily equivalent, implying
\begin{equation}
  \Delta S_A(s)=\Delta S_B(s).
\end{equation}
Together with Eq.~\eqref{eq:main-balance}, this gives
\begin{equation}
  \Delta S_A(s)=\Delta S_B(s)
  =\frac{1}{2}I(A:B)_{\omega(s)}\geq 0.
\end{equation}

For completeness, we now give the intrinsic continuum proof.
Let
$M_A=\mathcal A(A)$, $M_B=\mathcal A(B)$, and $M=M_A\vee M_B$.  For
disjoint closures the split property supplies the spatial identification
$M\simeq M_A\bar\otimes M_B$ and hence the faithful product state
$\sigma=\tau_A\otimes\tau_B$, where $\tau$ is the vacuum restricted to
$M$ \cite{Doplicher:1984split,Morinelli:2016split,Longo:2017lct}.  The
intrinsic mutual information and modular orbit are
\begin{align}
I_\varphi(A:B)&:=D_M(\varphi\Vert\varphi_A\otimes\varphi_B),
 \label{eq:main-araki-mi}\\
 \omega_s&:=\sigma\circ\vartheta_{-s}^{\tau}.
 \label{eq:main-algebraic-orbit}
\end{align}
where $\vartheta_s^\tau$ is the modular automorphism group and $D_M$ is
Araki relative entropy \cite{Araki:1975rel}.  The split property is used only
to define the product state. We assume that $D_M(\omega_s\Vert\sigma)<\infty$ for the CFT under
consideration,
we therefore define the algebraic MES, by
\begin{equation}
 \widetilde{\Delta\mathcal S}_X(s)
 :=\frac12D_M(\omega_s\Vert\sigma)
 -D_{M_X}(\omega_{X,s}\Vert\tau_X).
 \label{eq:main-algebraic-mes}
\end{equation}
\begin{theorem}[Algebraic MES under conformal exchange]
For the vacuum of a two-dimensional CFT satisfying the split property,
suppose that the interval-exchange involution $h$ is unitarily implemented
and leaves the vacuum invariant. Then, for every $s\in\mathbb{R}$, the two
algebraic MES quantities satisfy
\begin{equation}
  \widetilde{\Delta\mathcal S}_A(s)
  =
  \widetilde{\Delta\mathcal S}_B(s)
  =
  \frac{1}{2}I_{\omega_s}(A:B)
  \geq 0.
  \label{eq:main-cft-positivity}
\end{equation}
\end{theorem}

\begin{proof}
Let $\alpha_h:=\operatorname{Ad}_{U(h)}$ denote the automorphism induced by
the conformal involution $h$, with
\begin{equation}
  \alpha_h(x)=U(h)xU(h)^*,
  \qquad x\in M_{A\cup B}.
\end{equation}
M\"obius covariance implies that $\alpha_h$
exchanges $M_A$ and $M_B$ while preserving their union. Vacuum invariance,
together with the covariance of modular automorphisms, gives
\begin{equation}
  \tau\circ\alpha_h^{-1}=\tau,
  \qquad
  \alpha_h\circ\vartheta_t^\tau\circ\alpha_h^{-1}
  =
  \vartheta_t^{\,\tau\circ\alpha_h^{-1}}
  =
  \vartheta_t^\tau.
  \label{eq:main-modular-covariance}
\end{equation}
Since $\alpha_h$ exchanges the marginal states $\tau_A$ and $\tau_B$, it
also leaves their product state invariant:
\begin{equation}
  \sigma\circ\alpha_h^{-1}=\sigma.
\end{equation}
Using Eq.~\eqref{eq:main-algebraic-orbit} and
Eq.~\eqref{eq:main-modular-covariance}, we therefore obtain
\begin{align}
  \omega_s\circ\alpha_h^{-1}
  &=
  \sigma\circ\vartheta_{-s}^\tau\circ\alpha_h^{-1}
  \nonumber\\
  &=
  \sigma\circ\alpha_h^{-1}\circ\vartheta_{-s}^\tau
  =
  \omega_s.
  \label{eq:main-orbit-invariance}
\end{align}

The restriction
\begin{equation}
  \beta_h
  :=
  \left.\alpha_h^{-1}\right|_{M_A}
  :
  M_A\longrightarrow M_B
\end{equation}
is a normal isomorphism satisfying
\begin{equation}
  \tau_A=\tau_B\circ\beta_h,
  \qquad
  \omega_{A,s}=\omega_{B,s}\circ\beta_h.
\end{equation}
The invariance of Araki relative entropy under normal isomorphisms then
implies
\begin{equation}
  D_{M_A}\!\left(\omega_{A,s}\Vert\tau_A\right)
  =
  D_{M_B}\!\left(\omega_{B,s}\Vert\tau_B\right).
  \label{eq:main-local-relative-equality}
\end{equation}
By the definition of the algebraic MES, this equality gives
\begin{equation}
  \widetilde{\Delta\mathcal S}_A(s)
  =
  \widetilde{\Delta\mathcal S}_B(s).
\end{equation}
Since $\sigma=\tau_A\otimes\tau_B$, the chain rule for Araki relative
entropy gives
\begin{equation}
  D_M(\omega_s\Vert\sigma)
  =
  I_{\omega_s}(A:B)
  +D_{M_A}(\omega_{A,s}\Vert\tau_A)
  +D_{M_B}(\omega_{B,s}\Vert\tau_B).
\end{equation}
Therefore the sum of algebraic MES is
\begin{equation}
  \widetilde{\Delta\mathcal S}_A(s)
  +\widetilde{\Delta\mathcal S}_B(s)
  =
  I_{\omega_s}(A:B)\geq0.
\end{equation}
Then, each algebraic MES equals one half
of the mutual information, proving
Eq.~\eqref{eq:main-cft-positivity}.
\end{proof}

Note that $\widetilde{\Delta\mathcal S}_X$ denotes the algebraic, relative-entropy
quantity and should not, in general, be identified with the type-I entropy
shift
\begin{equation}
  \Delta S_X(s)
  :=
  S\!\left(\omega_X(s)\right)-S(\tau_X).
\end{equation}
Indeed, in a type-I realization the two definitions are related by
\begin{align}
  \widetilde{\Delta\mathcal S}_A(s)
  &=
  \Delta S_A(s)
  +
  \frac{1}{2}
  \left(
    \Delta\langle K_B\rangle_s
    -
    \Delta\langle K_A\rangle_s
  \right),
  \\
  \widetilde{\Delta\mathcal S}_B(s)
  &=
  \Delta S_B(s)
  +
  \frac{1}{2}
  \left(
    \Delta\langle K_A\rangle_s
    -
    \Delta\langle K_B\rangle_s
  \right).
\end{align}
The conformal involution symmetry enforces
\begin{equation}
  \Delta\langle K_A\rangle_s
  =
  \Delta\langle K_B\rangle_s,
\end{equation}
so that, for a conformally adapted type-I regulator,
\begin{equation}
  \widetilde{\Delta\mathcal S}_X(s)=\Delta S_X(s),
  \qquad X=A,B.
\end{equation}

\section{Connes cocycle interpretation}
\label{sec:main-cocycle}
The modular orbit in Eq.~\eqref{eq:main-algebraic-orbit} is precisely a
Connes-cocycle orbit.  To make this identification explicit, consider the
Connes cocycle derivative associated with the ordered pair $(\tau,\sigma)$,
\begin{equation}
 u_s \equiv [D\tau:D\sigma]_s =\Delta_{\tau\mid\sigma}^{\ii s}
   \Delta_\sigma^{-\ii s}
 \in M.
 \label{eq:main-intrinsic-cocycle}
\end{equation}
where $\Delta_{\tau\mid\sigma}$ is the relative modular operator associated
with the ordered pair $(\tau,\sigma)$, and
$\Delta_\sigma\equiv\Delta_{\sigma\mid\sigma}$ is the modular operator of $\sigma$. 
The Connes cocycle derivative satisfies the standard intertwining relation
\begin{equation}
 \vartheta_s^\tau
 =\operatorname{Ad}(u_s)\circ\vartheta_s^\sigma .
 \label{eq:main-cocycle-intertwining}
\end{equation}
Taking the inverse of this relation gives
\begin{equation}
 \vartheta_{-s}^\tau
 =\vartheta_{-s}^\sigma\circ\operatorname{Ad}(u_s^*).
 \label{eq:main-negative-intertwining}
\end{equation}
Since $\sigma$ is invariant under its own modular flow,
\begin{align}
 \omega_s
 &=\sigma\circ\vartheta_{-s}^\tau \nonumber\\
 &=\sigma\circ\vartheta_{-s}^\sigma
   \circ\operatorname{Ad}(u_s^*) \nonumber\\
 &=\sigma\circ\operatorname{Ad}(u_s^*).
 \label{eq:main-cocycle-orbit}
\end{align}
Thus the modular orbit is exactly the state orbit generated by the Connes
cocycle.

For comparison, in the finite-dimensional type-I setting, identifying
states with their density matrices, the same cocycle takes the form
\begin{equation}
 u_s
 =\tau_{AB}^{\ii s}\sigma_{AB}^{-\ii s}
 =e^{-\ii sK_\tau}
  e^{\ii s(K_A\otimes\id_B+\id_A\otimes K_B)} .
 \label{eq:main-cocycle-type-I}
\end{equation}
Its action on the product state is
\begin{align}
 u_s\sigma_{AB}u_s^*
 &=\tau_{AB}^{\ii s}\sigma_{AB}\tau_{AB}^{-\ii s}
 \nonumber\\
 &=e^{-\ii sK_\tau}\sigma_{AB}e^{\ii sK_\tau}
 =\omega_{AB}(s),
 \label{eq:main-cocycle-orbit}
\end{align}
which recovers Eq.~\eqref{eq:main-orbit}. 

\section{Discussion and outlook}
\label{sec:discussion}

We have investigated how the modular Hamiltonian of a correlated reference
state redistributes entropy when acting on the decorrelated product of its
marginals.  The resulting modular entanglement shifts need not have a
definite sign individually, but their sum is exactly the mutual information
generated along the orbit.  For two disjoint intervals in the vacuum of a
two-dimensional CFT, a conformal involution exchanging the intervals forces
the two algebraic shifts to coincide and hence to be separately
non-negative.  Finally, the orbit is identified intrinsically as the Connes
cocycle flow of the correlated reference state relative to the product of
its marginals.  

The appearance of mutual information in the balance law suggests a
connection with the perturbative emergence of nonlocal modular flow and
causal shadows.  Reference
\cite{Chen:2022cmr} analyzed the boundary shape response of mutual
information, while Ref.~\cite{Chen:2022eyi} related bulk mutual-information
corrections to nonlocal modular flow and the emergence of a perturbative
causal shadow.  Together with the leading small-cross-ratio modular
Hamiltonian of Ref.~\cite{Chandrasekaran:2021tkb}, it would be interesting to determine whether the resulting boundary
quantity is related to the displacement of the quantum extremal surface, the variation of generalized entropy, or another geometric measure of the perturbative causal shadow.

The Connes-cocycle interpretation opens a complementary holographic
direction.  A semiclassical bulk dual of Connes cocycle flow across a
boundary cut has been proposed in terms of a bulk kink transform
\cite{Bousso:2020yxi}.  The pair of states considered here, consisting of a correlated reference
state and the product of its marginals, differs from that setup, so the existing bulk construction does not apply directly. Determining
the bulk action of this particular cocycle may clarify whether the MES
captures a response of the entanglement wedge.  The proposed relation
between infinite-time cocycle flow and relational bulk reconstruction
\cite{Parrikar:2024rel} provides a further motivation for pursuing this
question.

\begin{acknowledgments}
L.C. thanks Huajia Wang for useful discussions at an early stage of this work.  L.C. is supported by the
Shuimu Tsinghua Scholar Program of Tsinghua University and the China
Postdoctoral Science Foundation under Grant No.~2025M783401. L.C. used OpenAI’s GPT-5.6 model to assist with technical derivations and manuscript preparation. All AI-assisted derivations and calculations were independently checked
by the author, who takes full
responsibility for the manuscript.
\end{acknowledgments}

\bibliography{ref}

\clearpage
\appendix
\setcounter{equation}{0}
\renewcommand{\theequation}{\Alph{section}\arabic{equation}}
\setcounter{figure}{0}
\renewcommand{\thefigure}{\Alph{section}\arabic{figure}}

\section{Quantum channels and modular orbits: a comparison}
\label{sec:background}

We compare the modular-flowed map with a fixed quantum channel, since the
entropy inequality for the latter motivates the search for an analogue of
the former.

A quantum channel is a completely positive trace-preserving (CPTP) linear map
that sends quantum states to quantum states.  In finite dimensions it has
the form
\begin{equation}
 \mathcal N_{A\to B}:\cL(\cH_A)\longrightarrow\cL(\cH_B).
\end{equation}
Consider now a channel on $A$ obtained by coupling $A$ to a fixed
environment state $\rho_B$ through a unitary $U$ on $\cH_A\otimes\cH_B$:
\begin{equation}
 \Phi_A(\rho_A)
 =\Tr_B\!\left[U(\rho_A\otimes\rho_B)U^\dagger\right].
 \label{eq:reduced-channel}
\end{equation}
We ask when
\begin{equation}
 S(\Phi_A(\rho_A))\geq S(\rho_A)
 \label{eq:channel-entropy-increase}
\end{equation}
holds for every input state $\rho_A$.  Recall that the quantum relative
entropy is
\begin{equation}
 D(\rho_A\Vert\sigma_A)
 =\Tr\!\left[\rho_A(\log\rho_A-\log\sigma_A)\right],
 \label{eq:relative-entropy}
\end{equation}
and that for the maximally mixed state $\pi_A=\id_A/d_A$ the entropy admits
the rewriting
\begin{equation}
 S(\rho_A)=\log d_A-D(\rho_A\Vert\pi_A).
 \label{eq:entropy-relative-max-mixed}
\end{equation}
If $\Phi_A$ is unital, so that $\Phi_A(\pi_A)=\pi_A$, the data-processing
inequality gives
$D(\Phi_A(\rho_A)\Vert\pi_A)\leq D(\rho_A\Vert\pi_A)$, which is equivalent
to \cref{eq:channel-entropy-increase}.  Conversely, applying
\cref{eq:channel-entropy-increase} to $\pi_A$ forces $\Phi_A(\pi_A)$ to have
the maximal entropy $\log d_A$.  Thus
\begin{equation}
 S(\Phi_A(\rho_A))\geq S(\rho_A)\ \forall\, \rho_A
 \;\Longleftrightarrow\;
 \Phi_A(\pi_A)=\pi_A.
 \label{eq:unital-characterization}
\end{equation}

The modular-flowed map of falls outside this
framework.  For fixed reference data $\tau_{AB}$ and modular time $s$,
the reduced map
\begin{equation}
 \Phi^{(\tau)}_{A,s}(\rho_A)
 =\Tr_B\!\left[
 \e^{-\ii sK_\tau}(\rho_A\otimes\tau_B)\e^{\ii sK_\tau}
 \right]
 \label{eq:fixed-modular-channel}
\end{equation}
is indeed CPTP, but changing the reference state $\tau_{AB}$ simultaneously
changes the input $\tau_A\otimes\tau_B$, the environment $\tau_B$, and the
generator $K_\tau$.  Hence the assignment from reference data to
modular-flowed data is not described by a single fixed CPTP channel.

\section{Modular BKM susceptibility}
The modular orbit has the expansion
\begin{equation}
 \begin{aligned}
  \omega_{AB}(s)&=\sigma_{AB}+sX+O(s^2),\\
  X&:=\dot\omega_{AB}(0)=-\ii[K_\tau,\sigma_{AB}].
 \end{aligned}
 \label{eq:short-time-orbit-tangent}
\end{equation}
Let
$X_A:=\Tr_BX$ and $X_B:=\Tr_AX$.  The product of the flowed marginals,
\begin{equation}
 \pi_{AB}(s):=\omega_A(s)\otimes\omega_B(s),
\end{equation}
starts from the same state and has the expansion
\begin{equation}
 \begin{aligned}
  \pi_{AB}(s)&=\sigma_{AB}+sX_{\mathrm{prod}}+O(s^2),\\
  X_{\mathrm{prod}}&:=X_A\otimes\tau_B+\tau_A\otimes X_B.
 \end{aligned}
 \label{eq:short-time-product-tangent}
\end{equation}
To expand the relative entropy, introduce the integral representation of the
matrix-logarithm derivative,
\begin{equation}
 \begin{aligned}
  \mathcal T_\rho(Z)
  &:=\left.\frac{\dd}{\dd\epsilon}
  \log(\rho+\epsilon Z)\right|_{\epsilon=0}\\
  &=\int_0^\infty (\rho+t\id)^{-1}Z(\rho+t\id)^{-1}\,\dd t,
 \end{aligned}
 \label{eq:bkm-log-derivative}
\end{equation}
which satisfies
$\Tr\!\left[\rho\,\mathcal T_\rho(Z)\right]=\Tr Z$.
For two normalized curves
$\rho(s)=\rho+sU+O(s^2)$ and $\eta(s)=\rho+sV+O(s^2)$ through the same
faithful state, direct differentiation of $D(\rho(s)\Vert\eta(s))$ gives
\begin{align}
 \left.\frac{\dd}{\dd s}
 D(\rho(s)\Vert\eta(s))\right|_{s=0}
 &=\Tr\!\left[
 \rho\bigl(\mathcal T_\rho(U)-\mathcal T_\rho(V)\bigr)
 \right]\nonumber\\
 &=\Tr(U-V)=0.
 \label{eq:relative-entropy-first-variation}
\end{align}
Because
$I(A:B)_{\omega(s)}=
 D\!\left(\omega_{AB}(s)\middle\Vert\pi_{AB}(s)\right)$,
\cref{eq:relative-entropy-first-variation} implies
\begin{equation}
  I'(0)=0.
 \label{eq:entropy-sum-first-derivative}
\end{equation}

In fact, the two local MES values are separately stationary.  For
subsystem $A$,
\begin{align}
 \Delta S_A'(0)
 &=-\Tr_A(X_A\log\tau_A)\nonumber\\
 &=\ii\Tr_{AB}\!\left(
 K_\tau[\sigma_{AB},\log\tau_A\otimes\id_B]\right)=0,
 \label{eq:local-first-derivative}
\end{align}
because $\sigma_{AB}=\tau_A\otimes\tau_B$ commutes with
$\log\tau_A\otimes\id_B$.  The same argument gives
\begin{equation}
 \Delta S_A'(0)=\Delta S_B'(0)=0.
 \label{eq:both-local-first-derivatives}
\end{equation}

For the same two curves through $\rho$, the direct second-order expansion is
\begin{equation}
 \begin{aligned}
  &D(\rho+sU+O(s^2)\Vert\rho+sV+O(s^2))\\
  &\qquad=\tfrac{s^2}{2}\,\Tr\!\left[
  (U-V)\mathcal T_\rho(U-V)\right]+O(s^3).
 \end{aligned}
 \label{eq:relative-entropy-second-variation}
\end{equation}
 Let us denote the eigenvalues of $\rho$ as $\{p_{j}\}$, then $\rho=\sum_jp_j|j\rangle\langle j|$ and
\begin{equation}
 \Tr\!\left[Z\mathcal T_\rho(Z)\right]
 =\sum_{j,k}
 \frac{\log p_j-\log p_k}{p_j-p_k}|Z_{jk}|^2\geq0,
 \label{eq:positive-quadratic-form}
\end{equation}
with the convention that the diagonal summands reduce to
$\sum_j |Z_{jj}|^2/p_j$ when $p_j=p_k$.  Applied to
$\omega_{AB}(s)$ and $\pi_{AB}(s)$ this gives
\begin{equation}
 I(A:B)_{\omega(s)}
 =\tfrac{s^2}{2}\,\chi_{\mathrm{mod}}(\tau)+O(s^3),
 \label{eq:bkm-mutual-information-expansion}
\end{equation}
where
\begin{equation}
 \chi_{\mathrm{mod}}(\tau)
 :=\Tr\!\left[
 X_{\mathrm{corr}}\,
 \mathcal T_{\sigma_{AB}}(X_{\mathrm{corr}})
 \right]
 \geq0.
 \label{eq:bkm-modular-susceptibility}
\end{equation}
\begin{equation}
 X_{\mathrm{corr}}
 :=X-X_{\mathrm{prod}}
 =X-X_A\otimes\tau_B-\tau_A\otimes X_B.
 \label{eq:bkm-correlation-tangent}
\end{equation}

Equivalently,
\begin{equation}
 \Delta S_A''(0)+\Delta S_B''(0)
 =\chi_{\mathrm{mod}}(\tau)\geq0.
 \label{eq:bkm-local-curvature-balance}
\end{equation}
Thus the perturbative result constrains the sum of the two local entropy
curvatures; it does not require the two individual second derivatives to
have the same sign.

For the two-qubit model,
$X_A=X_B=0$ and only the $(00,11)$ matrix elements contribute.  Therefore
\begin{equation}
 \chi_{\mathrm{mod}}(\tau)
 =8g^2(p_{00}-p_{11})\log\frac{p_{00}}{p_{11}},
 \label{eq:bkm-qubit-susceptibility}
\end{equation}
and
\begin{equation}
 I(A:B)_{\omega(s)}
 =4g^2(p_{00}-p_{11})\log\frac{p_{00}}{p_{11}}\,s^2
 +O(s^3).
 \label{eq:bkm-qubit-mutual-information}
\end{equation}

\section{Explicit two-qubit illustrations}
\label{sec:explicit-qubit}
The matrix form of $H$ is
\begin{equation}
 H(a,b,g)=
 \begin{pmatrix}
 u&0&0&2g\\
 0&v&0&0\\
 0&0&-v&0\\
 2g&0&0&-u
 \end{pmatrix},
 \label{eq:qubit-general-H}
\end{equation}
 where
\begin{equation}
 u=a+b,\quad v=a-b
 \label{eq:qubit-uvr}
\end{equation}
So only the even-parity subspace
$\operatorname{span}\{|00\rangle,|11\rangle\}$ is mixed.  Its block
\begin{equation}
    H_{\mathrm e}=\begin{pmatrix}
        u&2g\\2g&-u
    \end{pmatrix}
\end{equation}
satisfies $H_{\mathrm e}^2=r^2\id, \,\ r=\sqrt{u^2+4g^2}>0$, therefore, we have
\begin{equation}
 \e^{-H_{\mathrm e}}
 =\cosh r\,\id-\frac{\sinh r}{r}H_{\mathrm e}
 =\begin{pmatrix}
   A_r & C_r\\
   C_r & D_r
  \end{pmatrix},
 \label{eq:qubit-even-exponential}
\end{equation}
where
\begin{equation}
 \begin{aligned}
  A_r&=\cosh r-\frac{u}{r}\sinh r,\\
  C_r&=-\frac{2g}{r}\sinh r,\\
  D_r&=\cosh r+\frac{u}{r}\sinh r.
 \end{aligned}
\end{equation}
The reference state is obtained analytically as
\begin{align}
 \tau_{AB}&=\frac{1}{Z_\tau}
 \begin{pmatrix}
  A_r & 0 & 0 & C_r\\
  0 & \e^{-v} & 0 & 0\\
  0 & 0 & \e^{v} & 0\\
  C_r & 0 & 0 & D_r
 \end{pmatrix},
 \label{eq:qubit-general-tau}\\
 Z_\tau&=2(\cosh r+\cosh v). \nonumber
\end{align}
In particular, the modular Hamiltonian is 
\begin{equation}
 K_\tau=-\log\tau_{AB}=H+(\log Z_\tau)\id_{AB}.
 \label{eq:qubit-general-modular-Hamiltonian}
\end{equation}

The two reduced states are diagonal:
\begin{equation}
 \tau_A=\begin{pmatrix}\alpha&0\\0&1-\alpha\end{pmatrix},
 \qquad
 \tau_B=\begin{pmatrix}\beta&0\\0&1-\beta\end{pmatrix},
 \label{eq:qubit-general-marginals}
\end{equation}
with
\begin{equation}
 \alpha=\frac{A_r+\e^{-v}}{Z_\tau},
 \qquad
 \beta=\frac{A_r+\e^v}{Z_\tau}.
 \label{eq:qubit-alpha-beta}
\end{equation}
Consequently, the decorrelated input is
\begin{equation}
 \sigma_{AB}=\tau_A\otimes\tau_B
 =\operatorname{diag}(p_{00},p_{01},p_{10},p_{11}),
 \label{eq:qubit-general-sigma}
\end{equation}
with
\begin{equation}
 \begin{aligned}
  p_{00}&=\alpha\beta, &
  p_{01}&=\alpha(1-\beta),\\
  p_{10}&=(1-\alpha)\beta, &
  p_{11}&=(1-\alpha)(1-\beta).
 \end{aligned}
 \label{eq:qubit-product-populations}
\end{equation}

The scalar term $\log Z_{\tau} \id_{AB} $ in \cref{eq:qubit-general-modular-Hamiltonian} produces only
a global phase, so the density-matrix evolution can be calculated with
$\widetilde U_s=\e^{-\ii sH}$.  Define
\begin{equation}
 \begin{aligned}
  c_s&=\cos(rs)-\ii\frac{u}{r}\sin(rs),\\
  m_s&=-\ii\frac{2g}{r}\sin(rs),\\
  d_s&=c_s^*.
 \end{aligned}
\end{equation}
Then
\begin{equation}
 \widetilde U_s=
 \begin{pmatrix}
 c_s&0&0&m_s\\
 0&\e^{-\ii vs}&0&0\\
 0&0&\e^{\ii vs}&0\\
 m_s&0&0&d_s
 \end{pmatrix}.
 \label{eq:qubit-general-unitary}
\end{equation}
Since $\widetilde U_s$ couples only $|00\rangle$ and $|11\rangle$, direct multiplication gives 
\begin{equation}
 \omega_{AB}(s)=
 \begin{pmatrix}
  p_{00}+\delta(s)&0&0&\gamma(s)\\
  0&p_{01}&0&0\\
  0&0&p_{10}&0\\
  \gamma(s)^*&0&0&p_{11}-\delta(s)
 \end{pmatrix},
 \label{eq:qubit-general-omega}
\end{equation}
where $\delta(s)$ and $\gamma(s)$ are
\begin{equation}
 \begin{aligned}
  \delta(s)&=\frac{4g^2}{r^2}\sin^2(rs)\,(p_{11}-p_{00}), \\
  \gamma(s)&=\ii\frac{2g}{r}\sin(rs)\,c_s\,(p_{00}-p_{11}).
 \end{aligned}
\end{equation}
The odd-parity parameters $p_{01}$ and $p_{10}$ remain invariant along the orbit. Taking the partial traces yields
\begin{equation}
 \begin{aligned}
  \omega_A(s)&=\operatorname{diag}(\alpha+\delta(s),1-\alpha-\delta(s)),\\
  \omega_B(s)&=\operatorname{diag}(\beta+\delta(s),1-\beta-\delta(s)).
 \end{aligned}
 \label{eq:qubit-general-flowed-marginals}
\end{equation}

For a diagonal qubit state with eigenvalues $x$ and $1-x$, set
$h_2(x)=-x\log x-(1-x)\log(1-x)$.  The two MES values are therefore
\begin{equation}
 \begin{aligned}
  \Delta S_A(s)&=h_2(\alpha+\delta(s))-h_2(\alpha),\\
  \Delta S_B(s)&=h_2(\beta+\delta(s))-h_2(\beta).
 \end{aligned}
 \label{eq:qubit-general-entropy-changes}
\end{equation}

\section{The Connes cocycle derivative}
\label{app:connes-derivative}

In this appendix, we briefly review the Connes cocycle derivative and fix
the convention used in the main text.  For two faithful normal states
$\tau$ and $\sigma$ on $M$, the notation
\begin{equation}
 [D\tau:D\sigma]_s
\end{equation}
must be read as a single symbol: it denotes the Connes derivative of
$\tau$ with respect to $\sigma$.  The first entry, $\tau$, plays the role
of the numerator state, while the second entry, $\sigma$, is the reference
or denominator state.  

This notation is the noncommutative analogue of the Radon--Nikodym
derivative $d\tau/d\sigma$.  For example, if $M$ is commutative and the
two states are represented by probability densities $p$ and $q$, then
\begin{equation}
 [D\tau:D\sigma]_s
 =\left(\frac{p}{q}\right)^{\ii s}.
 \label{eq:app-cocycle-commutative}
\end{equation}
Thus the ordering of $\tau$ and $\sigma$ is essential.

To give the intrinsic definition, choose the standard representation of
$M$ and let $\Omega_\tau$ and $\Omega_\sigma$ be the natural-cone vectors
representing the two states.  The relative Tomita operator is defined on
$M\Omega_\sigma$ by
\begin{equation}
 S_{\tau\mid\sigma}(x\Omega_\sigma)
 =x^*\Omega_\tau ,
 \qquad x\in M,
 \label{eq:app-relative-tomita}
\end{equation}
and the corresponding relative modular operator is
\begin{equation}
 \Delta_{\tau\mid\sigma}
 =S_{\tau\mid\sigma}^*S_{\tau\mid\sigma}.
 \label{eq:app-relative-modular}
\end{equation}
Writing $\Delta_\sigma=\Delta_{\sigma\mid\sigma}$, the Connes cocycle
derivative is the strongly continuous unitary family
\begin{equation}
 [D\tau:D\sigma]_s
 :=\Delta_{\tau\mid\sigma}^{\ii s}
   \Delta_\sigma^{-\ii s}
 \in M.
 \label{eq:app-connes-definition}
\end{equation}
The cocycle derivative is characterized by the intertwining relation
\begin{equation}
 \vartheta_s^\tau(x)
 =[D\tau:D\sigma]_s\,
  \vartheta_s^\sigma(x)\,
  [D\tau:D\sigma]_s^*
 \label{eq:app-cocycle-intertwining}
\end{equation}
and the cocycle identity
\begin{equation}
 [D\tau:D\sigma]_{s+t}
 =[D\tau:D\sigma]_s\,
  \vartheta_s^\sigma\!
  \left([D\tau:D\sigma]_t\right).
 \label{eq:app-cocycle-law}
\end{equation}
The second identity also explains why the cocycle is generally not a
one-parameter group.

Setting $t=-s$ in Eq.~\eqref{eq:app-cocycle-law} gives
\begin{equation}
 [D\tau:D\sigma]_{-s}
 =\vartheta_{-s}^\sigma
  \left([D\tau:D\sigma]_s^*\right).
 \label{eq:app-cocycle-inverse}
\end{equation}

Finally, for a type-I algebra with density matrices $\rho_\tau$ and
$\rho_\sigma$, the abstract definition reduces to
\begin{equation}
 [D\tau:D\sigma]_s
 =\rho_\tau^{\ii s}\rho_\sigma^{-\ii s}.
 \label{eq:app-cocycle-type-I}
\end{equation}
This formula makes the ordering transparent.  Notice, however, that when
the two density matrices do not commute,
\begin{equation}
 \rho_\tau^{\ii s}\rho_\sigma^{-\ii s}
 \neq
 \left(\rho_\tau\rho_\sigma^{-1}\right)^{\ii s}
\end{equation}
in general.  Equation~\eqref{eq:app-cocycle-type-I} is therefore the
noncommutative replacement of the classical ratio
$(p/q)^{\ii s}$.

\end{document}